\documentclass[sigconf]{acmart}

\AtBeginDocument{%
  }

\setcopyright{none}
\renewcommand\footnotetextcopyrightpermission[1]{}
\copyrightyear{2026}
\acmYear{2026}
\acmConference[MobiHoc '26]{The 27th International Symposium on Theory,
  Algorithmic Foundations, and Protocol Design for Mobile Networks and
  Mobile Computing}{November 23--26, 2026}{Tokyo, Japan}
\acmISBN{978-1-4503-XXXX-X/2026/11}

\usepackage{amsmath,amssymb,bm}
\usepackage{booktabs}

\AtEndPreamble{%
  \theoremstyle{acmdefinition}%
  \newtheorem{assumption}[theorem]{Assumption}%
  \newtheorem{remark}[theorem]{Remark}%
  \theoremstyle{acmplain}%
}

\begin{document}

\title{Resource-Constrained Semantic-Aware Remote Estimation with Overlapping Sensor Coverage}

\author{Bowen Sun}
\affiliation{%
\institution{Link\"oping University}
\city{Link\"oping}
\country{Sweden}}
\email{bowen.sun@liu.se}

\author{Nikolaos Pappas}
\affiliation{%
\institution{Link\"oping University}
\city{Link\"oping}
\country{Sweden}}
\email{nikolaos.pappas@liu.se}

\renewcommand{\shortauthors}{Sun et al.}

\begin{abstract}

We study semantic-aware remote estimation of multiple finite-state Markov sources observed by K sensors with overlapping coverage. The sensors share a time-division multiple-access uplink and differ in transmission reliability, delivery delay, and transmission budget. In each slot, the scheduler jointly selects a source and one of its monitoring sensors, or remains idle, to minimize the long-run average cost of actuation error subject to global and per-sensor transmission-frequency constraints. We formulate this problem as a finite average-cost constrained Markov decision process. We show that the transmission resource functions have rank at most K, although the global constraint may still restrict the feasible region. Consequently, the Lagrangian depends only on K effective transmission costs, and an optimal constrained solution can be represented using at most K+1 deterministic policy–recurrent-class components. We further characterize the piecewise-affine concave Lagrangian value function and derive projected dual subgradient ascent over an explicitly bounded multiplier set. Numerical results illustrate the value-function structure, the need for policy randomization in a representative instance, and the interaction between global and per-sensor transmission budgets.
\end{abstract}

\keywords{Semantic communication, constrained Markov decision processes, sensor scheduling, Lagrangian duality}

\maketitle

\section{Introduction}

Remote estimation is a key function in networked control systems, cyber-physical systems, and autonomous platforms. In these applications, sensors transmit measurements to a remote estimator or controller over communication links with limited bandwidth, nonzero delay, packet losses, and energy constraints. Continuous reporting may be infeasible or inefficient, requiring a scheduler to determine which information should be transmitted and when.


The age of information (AoI)~\cite{Kaul2012AoI} measures the freshness of the information available at the receiver, while estimation-distortion metrics quantify the discrepancy between the source state and its reconstruction~\cite{Sun2025AoIMultiprocess}. Conventional forms of these metrics, however, may not distinguish the operational consequences of different state mismatches. In a control system, confusing two safe operating modes may have little effect, whereas estimating a hazardous state as safe can be costly even when the estimate is recent. Semantic-aware communication addresses this by assigning value to information according to its relevance to the underlying task~\cite{Kountouris2021SemEmpComm,Pappas2021GoalOriented,Salimnejad2024RemoteAct}.

The cost of actuation error (CAE) scores a state mismatch by the cost it induces at the actuator~\cite{Pappas2021GoalOriented}. Luo and Pappas~\cite{LuoPappas2024SemAware} studied CAE scheduling for multiple Markov sources sharing a single lossy transmitter under a transmission-frequency constraint. In that setting, each scheduling decision selects only the source to be updated. Many sensor networks have a more general coverage structure: a source may be observed by several sensors, and an individual sensor may observe several sources~\cite{Kalor2023MultiSensor}. The sensors covering the same source may differ in transmission reliability, delivery delay, and available communication resources, and these attributes need not be aligned. The link with the higher success probability may also be the one with the longer delay~\cite{Pan2021HeteroChannels}, so link quality cannot be summarized by a single scalar. A scheduling decision must therefore specify the source to update and the sensor through which the update is sent, which makes each slot an assignment over source--sensor pairs rather than a choice of source~\cite{Cosandal2025SensorSelect}.

Overlapping coverage creates a coupling between semantic importance and communication quality. The preferred sensor for a source depends not only on link reliability and delay, but also on the urgency and potential actuation cost of the source's current state. Moreover, a shared time-division multiple-access (TDMA) uplink couples all source--sensor pairs, while sensor-specific transmission budgets introduce additional long-run resource constraints. This coupling is a property of the access scheme: under multi-packet reception several sensors can be decoded in the same slot, and the corresponding CAE analysis takes a different form~\cite{Elessawy2026MPR}. Thus, source selection and sensor selection must be optimized jointly.

We study semantic-aware scheduling of multiple finite-state Markov sources observed by $K$ sensors with overlapping coverage. In each slot, a central scheduler selects at most one source--sensor pair for transmission, or leaves the channel idle. The objective is to minimize the long-run average CAE subject to a global transmission-frequency constraint and a transmission-frequency constraint per sensor. We formulate the problem as a finite average-cost constrained Markov decision process (CMDP) and characterize the structure of its primal and dual solutions. The structural results provide a reference for the development and evaluation of scalable approximations. Our main contributions are as follows.
\begin{itemize}
\item We formulate joint source--sensor scheduling under overlapping coverage as an average-cost CMDP. The model captures a shared TDMA uplink, sensor-dependent transmission reliability and delivery delay, a global transmission-frequency budget, and individual sensor budgets (Section~\ref{sec:system}).

\item We show that the $K{+}1$ constraints are linearly dependent, because at most one sensor transmits per slot. Three consequences follow: the Lagrangian cost depends on the multipliers only through the $K$ effective transmission costs $\mu_k=\lambda_0+\lambda_k$; the dual is exactly affine along $(1,-1,\dots,-1)$, so every dual maximizer lies on the boundary of $\mathbb{R}_+^{K+1}$; and the number of deterministic policies needed at the optimum drops from $K{+}2$ to $K{+}1$ (Proposition~\ref{prop:effective-price}).


\item Using occupation measures, we prove that the constrained optimum is attained and is independent of the initial state. We further show that an optimal solution can be represented by a mixture of at most $K{+}1$ deterministic policy--recurrent-class components (Theorem~\ref{thm:mixture}). The analysis explicitly accommodates multichain policies, which may arise when a deterministic scheduler leaves some sources unupdated indefinitely.

\item We prove that the inner Lagrangian value is piecewise-linear and concave in the multiplier vector, with supergradients equal to the transmission-frequency vectors of the inner-optimal policies (Theorem~\ref{thm:pwlc}). With an explicit multiplier bound from a strictly feasible policy (Lemma~\ref{lem:lambda-bound}), projected dual subgradient ascent then applies with an $O(1/\sqrt{N})$ guarantee (Proposition~\ref{prop:dsg}).

\item We evaluate instances small enough for exact occupation-measure linear programming (Section~\ref{sec:numerical}), and confirm the piecewise-linear value with matching frequency plateaus, a budget at which randomization is required, and the predicted exchange between the global and per-sensor multipliers.
\end{itemize}

\subsection{Related work}
\emph{Freshness and mismatch metrics.} AoI scheduling has been studied through throughput-optimal and index policies for broadcast networks~\cite{Kadota2018AoIScheduling} and Whittle indices for AoI and general age penalties~\cite{Maatouk2021WhittleAoI,Tripathi2024WhittleAoI}. Age of incorrect information penalizes how long the receiver's estimate stays wrong~\cite{Maatouk2020AoII,Maatouk2023AoII}, with Whittle-type policies for Markov sources in~\cite{Kriouile2021AoIIMarkov} and extensions to channels with random delay in~\cite{Chen2024AoIIDelay}. CAE is more specific than both: its asymmetric cost matrix records which mismatch occurred~\cite{Pappas2021GoalOriented,LuoPappas2024SemAware,Luo2026DataSignif}. Multi-process sensor scheduling has also been studied for linear systems under covariance-based error~\cite{Han2017SensorScheduling,Sun2025AoIMultiprocess}, and the closest multi-channel models are~\cite{Ornee2023WhittleGM,Zhou2025AgeMultiChannel}. We differ from these works in the model: finite-state Markov sources, asymmetric semantic costs, an arbitrary bipartite sensor--source coverage graph, and both a global and a per-sensor budget.

\noindent\emph{Constrained MDPs and duality.} Our analysis uses occupation measures for constrained MDPs~\cite{Altman1999CMDP,HordijkKallenberg1984}, which underpin both classical randomization results~\cite{Ross1989MultiConstraint} and recent first-order solvers~\cite{Grontas2026OperatorSplitting}. Constrained scheduling for multi-source status updating has likewise been posed as a CMDP~\cite{Zakeri2024AoICMDP}. The single-transmitter CAE problem of~\cite{LuoPappas2024SemAware} has one constraint and hence a scalar multiplier, for which the Lagrangian value is a one-dimensional piecewise-linear concave curve and a line search suffices. With $K$ sensors the multiplier is a vector and the curve becomes a polyhedral surface, so a search along single coordinates is no longer sufficient. The dependence between the global and the per-sensor constraints has no counterpart in the single-transmitter model, and it is what keeps the effective dimension at $K$.

\noindent\emph{Scalable and online scheduling.} Restless bandit relaxations~\cite{Whittle1988RestlessBandits}, weakly coupled dynamic programs~\cite{Adelman2008Weakly}, multi-action indices~\cite{Hodge2015MultiAction}, and Lyapunov drift-plus-penalty control~\cite{Neely2010Stochastic} are the standard routes to large instances. All of them approximate the CMDP studied here, and the results of Section~\ref{sec:structure} give the exact reference against which such approximations must be measured. Developing them for overlapping coverage is left to future work.

\noindent\emph{Organization.} Section~\ref{sec:system} gives the model and the CMDP formulation, Section~\ref{sec:structure} the structural results, and Section~\ref{sec:numerical} the numerical results. Section~\ref{sec:conclusion} concludes. Proof sketches are given inline; full proofs are in the extended version.

\section{System Model and Problem Formulation}\label{sec:system}

\subsection{Network and sensing graph}

As shown in Fig.~\ref{fig:system}, we consider a system which operates in discrete time, with slots indexed by $t\in\mathbb{N}_0$. Let $\mathcal{M}=\{1,\ldots,M\}$ and $\mathcal{K}=\{1,\ldots,K\}$ denote the sets of sources and sensors, respectively. Each sensor $k\in\mathcal{K}$ observes a nonempty subset $\mathcal{M}_k\subseteq\mathcal{M}$. The coverage sets may overlap, and we assume that $\bigcup_{k\in\mathcal{K}}\mathcal{M}_k=\mathcal{M}$, so that every source is observed by at least one sensor. For each source $m\in\mathcal{M}$, define $\mathcal{K}_m\triangleq\{k\in\mathcal{K}:m\in\mathcal{M}_k\}$ as the set of sensors that can observe source $m$. The sensing relation is represented by the bipartite graph $\mathcal{G}=(\mathcal{K},\mathcal{M},\mathcal{E})$, whose edge set is
\begin{equation}\label{eq:sensing-graph}
\mathcal{E}\triangleq
\{(k,m)\in\mathcal{K}\times\mathcal{M}:m\in\mathcal{M}_k\}.
\end{equation}
When $|\mathcal{K}_m|>1$, source $m$ can be monitored and updated through multiple sensors, which may differ in their link reliabilities and remaining transmission budgets.

\begin{figure}[htp]
\centering
\includegraphics[width=0.8\columnwidth]{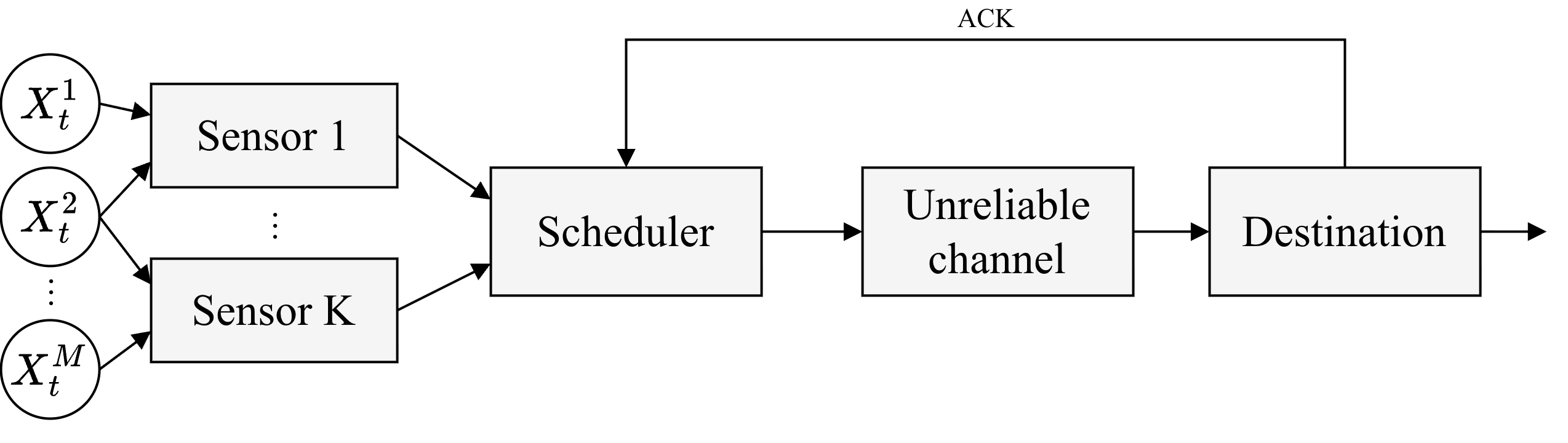}
\caption{System architecture}
\label{fig:system}
\end{figure}

Each source $m\in\mathcal{M}$ evolves as a finite-state, time-homogeneous Markov chain $\{X_t^m\}$ on $\mathcal{X}^m=\{1,\ldots,N_m\}$, with transition matrix $Q^m$, where $Q_{ij}^m=\Pr(X_{t+1}^m=j\mid X_t^m=i)$. We assume that each source chain is irreducible and aperiodic and that the $M$ source chains evolve independently. Consequently, their product chain is also irreducible and aperiodic. We call $Q^m$ slowly evolving if it is diagonally dominant and rapidly evolving otherwise.

A central scheduler activates at most one sensing edge in each slot. Its action is
\begin{equation}\label{eq:action}
A_t = (k_t, m_t) \in \mathcal{A} := \{(0,0)\}\cup\mathcal{E},
\end{equation}
where $A_t=(0,0)$ denotes an idle slot. If $A_t=(k,m)\in\mathcal{E}$, sensor $k$ samples the current state $X_t^m$ of source $m$ and transmits it over its uplink to the destination. 

Let $H_t^k\in\{0,1\}$ denote the transmission outcome on sensor $k$’s uplink in slot $t$, where $H_t^k=1$ indicates a successful transmission. For each $k\in\mathcal{K}$, $\{H_t^k\}_t$ is an i.i.d. Bernoulli process satisfying
$\Pr(H_t^k=1)=p_s^k$, $p_s^k\in(0,1]$. The channel outcomes are independent across sensors and slots and are independent of the source processes. Sensor $k$’s uplink has a fixed delay $d_k\in\{0,1\}$ slots. The transition kernels $\{Q^m\}$ and delays $\{d_k\}$ are known. Acknowledgements are immediate and error-free.

We adopt a centralized full-information model in which the scheduler obtains the current source states through a reliable local control plane and tracks the receiver estimates using error-free acknowledgements.
\begin{assumption}\label{ass:info}
At the start of every slot $t$, before selecting $A_t$, the scheduler knows the current source-states $\{X^m_t\}_{m\in\mathcal{M}}$, the receiver estimates $\{\hat X^m_t\}_{m\in\mathcal{M}}$, and all past actions and acknowledgements.
\end{assumption}

\subsection{Reconstruction and cost of actuation error}

If $A_t=(k,m)$ and the transmission succeeds, the receiver
replaces the estimate of source $m$ with the sampled value;
otherwise, it retains the previous estimate. Thus
\begin{equation}\label{eq:recon-next}
\hat X_{t+1}^n =
\begin{cases}
X_t^m, & \text{if } A_t=(k,m),\ n=m,\ H_t^k=1,\\
\hat X_t^n, & \text{otherwise}.
\end{cases}
\end{equation}
The delay $d_k$ determines whether a successfully transmitted
update can be used in the current slot. Let $\tilde{X}_t^n$
denote the estimate available to the actuator in slot $t$, given by
\begin{equation}\label{eq:actuation-estimate}
\tilde X_t^n =
\begin{cases}
X_t^m, & \text{if } A_t=(k,m),\ n=m,\ d_k=0,\ H_t^k=1,\\
\hat X_t^n, & \text{otherwise}.
\end{cases}
\end{equation}
Following~\cite{Pappas2021GoalOriented,LuoPappas2024SemAware}, for each source $m$, we use a state-dependent and generally asymmetric per-source cost $\delta^m:\mathcal{X}^m\times\mathcal{X}^m\to\mathbb{R}_{\ge0}$ satisfying $\delta^m(i,i)=0$ and in general $\delta^m_{i,j}\neq\delta^m_{j,i}$. The one-slot expected cost of actuation error under action $a$ is,
\begin{equation}\label{eq:cae}
c(s,a):=\mathbb{E}\Bigl[\textstyle\sum_{m\in\mathcal{M}}\omega_m\,
\delta^m\bigl(X_t^m,\tilde X_t^m(a,H_t)\bigr)\ \Big|\ S_t=s,A_t=a\Bigr],
\end{equation}
where $S_t$ denotes the joint system state defined in
Section~\ref{sec:stb}, and $\omega_m\geq 0$ is a weight capturing the importance of source $m$.. Charging the cost at the actuation epoch of the current slot keeps a single accounting epoch per slot even though sensors have different delays.

\subsection{State, transitions, and budgets}\label{sec:stb}

For each source $m$ let $S^m_t=(X^m_t,\hat X^m_t)\in\mathcal{S}^m:= \mathcal{X}^m\times\mathcal{X}^m$. The joint state is $S_t=(S^m_t)_{m}\in\mathcal{S}:=\prod_m\mathcal{S}^m$, so $|\mathcal{S}|=\prod_m N_m^2$. The sources evolve independently, whereas only the receiver estimate of the selected source can change. Given $a=(k,m)$, the estimate transition of source $m$ depends on the sensor through $p^k_s$, while $d_k$ affects the current-slot cost through~\eqref{eq:actuation-estimate}. The kernel factorizes as $ P(s'\mid s, a) = \prod_{m\in\mathcal{M}} P^m\bigl(s'^{(m)}\mid s^{(m)},a\bigr)$, where $P^m$ denotes the transition kernel of the local
state $(X_t^m,\hat{X}_t^m)$. The source state evolves
according to $Q^m$. If $a=(k,m)$, the receiver estimate
is updated to the sampled state with probability $p_s^k$
and remains unchanged with probability $1-p_s^k$,
as specified in \eqref{eq:recon-next}. If source $m$ is not selected, its receiver estimate remains unchanged.

Define the global and per-sensor transmission indicators
\begin{equation}\label{eq:indicators}
f_0(A_t) = \mathbf{1}\{A_t\neq(0,0)\}, \quad
f_k(A_t) = \mathbf{1}\{k_t=k\},\ k\in\mathcal{K},
\end{equation}
and let $\bar F_{j,s}(\pi):=\limsup_{T\to\infty}\frac1T\sum_{t=1}^{T} \mathbb{E}^\pi[f_j(A_t)\mid S_1=s]$. The system must satisfy a global budget $F_{\max}\in(0,1]$ and a per-sensor budget $F_{\max,k}\in(0,1]$
\begin{equation}\label{eq:freq}
\bar F_{0,s}(\pi)\le F_{\max}, \qquad
\bar F_{k,s}(\pi) \le F_{\max,k}\quad\forall k\in\mathcal{K}.
\end{equation}
The action set~\eqref{eq:action} already enforces the per-slot TDMA constraint $\sum_k f_k(A_t)\le1$. Moreover, the transmission indicators satisfy $f_0(a)=\sum_{k\in\mathcal{K}}f_k(a)$. Thus, the $K+1$ transmission-indicator functions are linearly dependent and span a space of dimension at most $K$.

\subsection{CMDP formulation}\label{sec:formulation}

Under Assumption~\ref{ass:info}, the system state $S_t$ is fully observed. Let $I_t=(S_1,\ldots,S_t,A_1,\ldots,A_{t-1})$ denote the observed state--action history available before selecting $A_t$. A randomized policy $\pi=\{\pi_t\}_{t\geq 1}$ specifies an action distribution $\pi_t(\cdot\mid I_t)\in\Delta(\mathcal{A})$ at each slot, where $\Delta(\mathcal{A})$ denotes the set of probability distributions over $\mathcal{A}$. Let $\Pi$ denote the class of all such policies. Let $\Pi_{\mathrm{MD}}$ and $\Pi_{\mathrm{MR}}$ denote the classes of stationary Markov deterministic (MD) and stationary Markov randomized (MR) policies, respectively, with $\Pi_{\mathrm{MD}}\subseteq\Pi_{\mathrm{MR}}\subseteq\Pi$.
\begin{definition}\label{def:cmdp}
The hierarchical TDMA scheduling problem is
\begin{equation}\label{eq:CMDP}
\begin{aligned}
\mathcal{P}_{\mathrm{CMDP}}:\quad \min_{\pi\in\Pi}\ &\bar C_s(\pi)\\
\text{s.t.}\quad &\bar F_{0,s}(\pi) \le F_{\max},\\
&\bar F_{k,s}(\pi) \le F_{\max,k}, \quad \forall k\in\mathcal{K}.
\end{aligned}
\end{equation}
\end{definition}

We write $C^*$ for the optimal value of~\eqref{eq:CMDP} and set $F_{\max,0}:=F_{\max}$.

\section{Structure of the Constrained Optimum}\label{sec:structure}

We do not assume any threshold or monotonicity structure for the optimal joint policy: the CAE can be asymmetric and the coverage sets can overlap, and neither property survives in general. What we do establish is the structure that holds regardless: existence and attainment of the optimum, how many deterministic policies a randomized optimal policy needs, the shape of the dual, and how many independent transmission costs the problem actually has.

\subsection{Strict feasibility and reachability}

The positive budgets give a strictly feasible policy directly. Pick $\epsilon\in(0,\min\{F_{\max}/K,\min_kF_{\max,k}\})$, set $\alpha_k=\epsilon$, and choose
\begin{equation}\label{eq:slater-margin}
0<\sigma\le \min\bigl\{F_{\max}-K\epsilon,\ \min_k(F_{\max,k}-\epsilon)\bigr\}.
\end{equation}
The state-independent policy $\pi^{\mathrm{SA}}$ activates sensor $k$ with probability $\alpha_k$, picks a source uniformly from $\mathcal{M}_k$, and idles otherwise.

\begin{lemma}\label{lem:slater}
The policy $\pi^{\mathrm{SA}}$ satisfies $\bar F_j(\pi^{\mathrm{SA}})\le F_{\max,j}-\sigma$ for $j=0,\dots,K$. Moreover, under the model assumptions of Section~\ref{sec:system}, the finite MDP on $\mathcal{S}$ is communicating in the sense of~\cite[Sec.~8.3]{Puterman1994}, and the chain induced by $\pi^{\mathrm{SA}}$ is irreducible and aperiodic.
\end{lemma}

\begin{proof}[Proof sketch]
The frequencies follow from $\bar F_0(\pi^{\mathrm{SA}})=\sum_k\alpha_k$, $\bar F_k(\pi^{\mathrm{SA}})=\alpha_k$, and~\eqref{eq:slater-margin}. For reachability, drive the estimates to their targets one at a time: wait until source $m$ visits the target estimate value, then activate any covering edge, which has positive probability because $\pi^{\mathrm{SA}}$ uses every edge and $p^k_s>0$. Estimates already set are not disturbed while other sources are updated. Then idle until the product source chain reaches the target true state, again with positive probability. Since idling has positive probability, the aperiodic return paths of the product source chain are return paths of the joint chain.
\end{proof}

Strict feasibility is what the dual analysis needs; the communicating property is what makes the average-cost results independent of the initial state.

\subsection{Lagrangian relaxation and existence}

We dualize the $K{+}1$ constraints with $\boldsymbol\lambda=(\lambda_0,\lambda_1,\dots,\lambda_K)\in\mathbb{R}_+^{K+1}$. The Lagrangian one-slot cost is
\begin{equation}\label{eq:lagcost}
\ell^{\boldsymbol\lambda}(s,a) = c(s,a) + \lambda_0 f_0(a)
+ \sum_{k=1}^{K}\lambda_k f_k(a),
\end{equation}
and the inner problem is $L^*(\boldsymbol\lambda):=\inf_{\pi\in\Pi}\bar L^{\boldsymbol\lambda}_s(\pi)$ with $\bar L^{\boldsymbol\lambda}_s(\pi)=\bar C_s(\pi)+\sum_j\lambda_j \bar F_{j,s}(\pi)$. The dual objective is
\begin{equation}\label{eq:dual-objective}
D(\boldsymbol\lambda)=L^*(\boldsymbol\lambda)
-\sum_{j=0}^{K}\lambda_jF_{\max,j},
\end{equation}
which lower-bounds $C^*$ for every $\boldsymbol\lambda\ge0$ by weak duality.

\begin{theorem}\label{thm:bellman}
For every $\boldsymbol\lambda\ge0$ the inner problem has a stationary deterministic optimal policy $\pi^*_{\boldsymbol\lambda}\in\Pi_{MD}$, the optimal gain $L^*(\boldsymbol\lambda)$ does not depend on the initial state, and there is a bias vector $\bm h^{\boldsymbol\lambda}$ with
\begin{equation}\label{eq:bellman}
L^*(\boldsymbol\lambda) + h^{\boldsymbol\lambda}(s)
= \min_{a\in\mathcal{A}}\Bigl\{\ell^{\boldsymbol\lambda}(s,a)
+ \sum_{s'}P(s'\mid s,a)h^{\boldsymbol\lambda}(s')\Bigr\}.
\end{equation}
\end{theorem}

\begin{proof}[Proof sketch]
$\mathcal{S}$ and $\mathcal{A}$ are finite and the CAE matrices are finite, so $\ell^{\boldsymbol\lambda}$ is bounded. By Lemma~\ref{lem:slater} the MDP is communicating, hence weakly communicating, so the optimal gain is constant in the initial state and the constant-gain optimality equation~\eqref{eq:bellman} has a solution whose minimizing selector is average-cost optimal~\cite[Ch.~8--9]{Puterman1994}. Multichain theory is needed because a fixed stationary policy may be multichain: a deterministic policy can stop updating a source, which freezes its estimate.
\end{proof}

\begin{remark}\label{rem:edge-gain}
Write $a_0:=(0,0)$. For an edge $(k,m)\in\mathcal{E}$ define
\begin{equation}\label{eq:edge-gain}
\begin{aligned}
\Gamma_{k,m}^{\boldsymbol\lambda}(s):={}&c(s,a_0)-c\bigl(s,(k,m)\bigr)
-(\lambda_0+\lambda_k)\\
&+\sum_{s'}\bigl[P(s'\mid s,a_0)-P(s'\mid s,(k,m))\bigr]
h^{\boldsymbol\lambda}(s').
\end{aligned}
\end{equation}
Equation~\eqref{eq:bellman} says that one optimal selector idles when $\max_{(k,m)}\\\Gamma^{\boldsymbol\lambda}_{k,m}(s)\le0$ and otherwise activates a maximizing edge, so the decision compares the immediate CAE reduction and the future bias reduction against the effective transmission cost $\mu_k=\lambda_0+\lambda_k$ of sensor $k$. This is exact, not a low-complexity rule: $h^{\boldsymbol\lambda}$ lives on the joint state space, and for a general asymmetric CAE we do not claim that $\Gamma^{\boldsymbol\lambda}_{k,m}$ is monotone in the state.
\end{remark}

\subsection{The dual is piecewise linear and concave}

Some stationary deterministic policies are multichain, so quantities such as $\bar C(\pi)$ can depend on the initial state. We therefore anchor all averages at an arbitrary fixed $s\in\mathcal{S}$; the Ces\`aro limits exist for every stationary policy on a finite chain.

\begin{theorem}\label{thm:pwlc}
For any fixed $s\in\mathcal{S}$, the inner value $L^*(\boldsymbol\lambda)$ is the lower envelope of finitely many affine functions of $\boldsymbol\lambda$,
\begin{equation}\label{eq:pwlc-envelope}
L^*(\boldsymbol\lambda)=\min_{\pi\in\Pi_{MD}}\Bigl\{
\bar C_s(\pi)+\textstyle\sum_{j=0}^{K}\lambda_j\bar F_{j,s}(\pi)\Bigr\},
\end{equation}
and the value of the envelope does not depend on $s$. Hence $L^*$ is concave, piecewise linear, and nondecreasing in every coordinate, and so is $D$ up to an affine term. Moreover
\begin{equation}\label{eq:superdiff}
\partial L^*(\boldsymbol\lambda)=\mathrm{conv}\Bigl\{
\bigl(\bar F_{0,s}(\pi),\dots,\bar F_{K,s}(\pi)\bigr):
\pi\ \text{active at }\boldsymbol\lambda\Bigr\}.
\end{equation}
\end{theorem}

\begin{proof}[Proof sketch]
$\Pi_{MD}$ is finite, and for each fixed $\pi\in\Pi_{MD}$ the anchored average $\bar L^{\boldsymbol\lambda}_s(\pi)$ is affine in $\boldsymbol\lambda$ and at least $L^*(\boldsymbol\lambda)$. Conversely Theorem~\ref{thm:bellman} supplies, for each $\boldsymbol\lambda$, a deterministic policy whose gain equals $L^*(\boldsymbol\lambda)$ from every initial state, so its affine piece attains the envelope for every anchor. A minimum of finitely many affine functions is concave, continuous, and piecewise linear, and it is nondecreasing in each coordinate because every $\bar F_{j,s}(\pi)\ge0$. Danskin's theorem gives the superdifferential as the convex hull of the slopes of the active pieces~\cite[Thm.~D.4.4.2]{HiriartUrruty2001Fundamentals}.
\end{proof}

Equation~\eqref{eq:superdiff} is the practical part of the theorem: a supergradient of the dual at $\boldsymbol\lambda$ is the vector of constraint violations of any inner-optimal deterministic policy, so one Bellman solve produces both the dual value and a search direction.

\subsection{Existence and the size of the optimal mixture}

The following statement is the main structural result. It uses the occupation-measure linear program for multichain average-cost models~\cite{HordijkKallenberg1984,Altman1999CMDP}. Because deterministic policies here can be multichain, the mixture is stated over policy--recurrent-class pairs: for $\pi\in\Pi_{MD}$ and a recurrent class $R$ of $\pi$, let $\bar C(\pi;R)$ and $\bar F_j(\pi;R)$ be the stationary averages of $\pi$ started inside $R$.

\begin{theorem}\label{thm:mixture}
The constrained optimum $C^*$ is independent of the initial state and is attained. There exist $\pi^*_i\in\Pi_{MD}$, recurrent classes $R_i$ of $\pi^*_i$, and weights $q_i\ge0$ with $\sum_{i=0}^{K}q_i=1$ such that
\begin{equation}\label{eq:mixture}
C^*=\sum_{i=0}^{K} q_i\,\bar C(\pi^*_i;R_i),\quad
\sum_{i=0}^{K} q_i\,\bar F_j(\pi^*_i;R_i)\le F_{\max,j}\ \ \forall j.
\end{equation}
An optimal policy is realized by a single randomization at $t=1$: with probability $q_i$, steer the chain into $R_i$, which is possible from any initial state because the MDP is communicating, and follow $\pi^*_i$ afterwards; the finite transient prefix does not change long-run averages. If the constraints binding at the optimum span a space of dimension $r\le K$, at most $r{+}1$ deterministic policies are needed, which sharpens the classical bound of one more than the number of constraints~\cite{Ross1989MultiConstraint}.
\end{theorem}

\begin{proof}[Proof sketch]
The set $\mathcal{X}$ of long-run state--action occupation measures of stationary policies is a polytope defined by balance and normalization constraints~\cite{HordijkKallenberg1984}, and both the average cost and the average frequencies are linear on it. Since the MDP is communicating (Lemma~\ref{lem:slater}), every $x\in\mathcal{X}$ is reachable from every initial state after a finite prefix, which does not affect Ces\`aro averages; hence $C^*$ does not depend on $s$ and~\eqref{eq:CMDP} is equivalent to the linear program $\min_{x\in\mathcal{X}}\langle c,x\rangle$ subject to $\langle f_j,x\rangle\le F_{\max,j}$. The program is feasible by Lemma~\ref{lem:slater} and attains its optimum at an extreme point $x^*$ of $\mathcal{X}\cap H$, where $H$ is the intersection of the constraint half-spaces. Every extreme point of $\mathcal{X}$ is the occupation measure of a stationary deterministic policy supported on one recurrent class~\cite{HordijkKallenberg1984}. If $r$ is the rank of the binding constraint functionals at $x^*$, then $x^*$ lies on a face of $\mathcal{X}$ of dimension at most $r$, since otherwise that face would contain a nonzero direction annihilating all binding functionals. Carath\'eodory's theorem~\cite[Thm.~17.1]{Rockafellar1970Convex} then writes $x^*$ as a combination of at most $r{+}1$ extreme points, and then we obtain $r\le K$.
\end{proof}

\subsection{Effective transmission costs and the degenerate dual direction}
\label{sec:dual-geometry}

The $K{+}1$ multipliers are not $K{+}1$ independent search directions, for the same reason.

\begin{proposition}\label{prop:effective-price}
Let $\bm d:=(1,-1,\dots,-1)\in\mathbb{R}^{K+1}$ and $\kappa:=\sum_{k}F_{\max,k}-F_{\max}$. Then:
\begin{enumerate}
\item The Lagrangian cost~\eqref{eq:lagcost} depends on $\boldsymbol\lambda$ only through the $K$ \emph{effective transmission costs} $\mu_k:=\lambda_0+\lambda_k$,
\begin{equation}\label{eq:effective-price}
\ell^{\boldsymbol\lambda}(s,a)=c(s,a)+\sum_{k\in\mathcal{K}}\mu_k f_k(a),
\end{equation}
so $L^*(\boldsymbol\lambda+t\bm d)=L^*(\boldsymbol\lambda)$ whenever both points are nonnegative.
\item $D$ is exactly affine along $\bm d$: $D(\boldsymbol\lambda+t\bm d)=D(\boldsymbol\lambda)+t\kappa$.
\item Every dual maximizer lies on the boundary of $\mathbb{R}_+^{K+1}$: if $\kappa>0$ then $\min_k\lambda^*_k=0$, and if $\kappa<0$ then $\lambda^*_0=0$.
\end{enumerate}
\end{proposition}

\begin{proof}[Proof sketch]
(1) Substitute $f_0(a)=\sum_kf_k(a)$  into~\eqref{eq:lagcost}; the shift $\boldsymbol\lambda\mapsto\boldsymbol\lambda+ t\bm d$ maps $\mu_k\mapsto(\lambda_0+t)+(\lambda_k-t)=\mu_k$, so the inner MDP is unchanged. (2) Combine (1) with~\eqref{eq:dual-objective}. (3) If $\kappa>0$ and $\tau:=\min_k\lambda^*_k>0$, then $\boldsymbol\lambda^*+\tau\bm d\ge0$ has a strictly larger dual value by (2), a contradiction; the case $\kappa<0$ is symmetric with the step $-\lambda^*_0\bm d$.
\end{proof}

About part~(3), when $\kappa<0$ the per-sensor budgets already imply the global budget, so the global multiplier is zero at optimality. When $\kappa>0$ at least one sensor budget is not binding. Removing this redundant direction costs no Bellman solve, but it does not make a coordinatewise search globally optimal, because $D$ is not differentiable in general.

\subsection{Computing the optimum}\label{sec:dsg}

Theorem~\ref{thm:pwlc} makes the dual a concave piecewise-linear maximization. To make it a maximization over a compact set, we bound the multipliers using the strict-feasibility margin of Lemma~\ref{lem:slater}.

\begin{lemma}\label{lem:lambda-bound}
Every dual maximizer satisfies
\begin{equation}\label{eq:lambda-bound}
\sum_{j=0}^{K}\lambda_j^*\ \le\
\frac{\bar C(\pi^{\mathrm{SA}})-C^*}{\sigma}\ \le\
\frac{\bar C(\pi^{\mathrm{SA}})}{\sigma}=:\Lambda_{\max},
\end{equation}
so the dual can be restricted to $\Lambda:=\{\boldsymbol\lambda\in\mathbb{R}_+^{K+1}: \|\boldsymbol\lambda\|_1\le\Lambda_{\max}\}$ without loss of optimality.
\end{lemma}

\begin{proof}[Proof sketch]
Evaluating the inner infimum at $\pi^{\mathrm{SA}}$ and using $\bar F_j(\pi^{\mathrm{SA}})\le F_{\max,j}-\sigma$ gives $D(\boldsymbol\lambda)\le\bar C(\pi^{\mathrm{SA}})-\sigma\|\boldsymbol\lambda\|_1$ for all $\boldsymbol\lambda\ge0$. Strong duality holds for the finite CMDP under strict feasibility~\cite{HordijkKallenberg1984,Altman1999CMDP}, so $D(\boldsymbol\lambda^*)=C^*$; substituting gives~\eqref{eq:lambda-bound}.
\end{proof}

At an iterate $\boldsymbol\lambda^n$, solve~\eqref{eq:bellman} to get $\pi^{\boldsymbol\lambda^n}\in\Pi_{MD}$ and a recurrent class $R^n$; the vector $\bm g^n=(\bar F_j(\pi^{\boldsymbol\lambda^n};R^n)-F_{\max,j})_{j=0}^{K}$ is a supergradient of $D$ by~\eqref{eq:superdiff}, and $\|\bm g^n\|_\infty\le1$. The projected step is $\boldsymbol\lambda^{n+1}=\Pi_\Lambda(\boldsymbol\lambda^n +\eta\bm g^n)$.

\begin{proposition}\label{prop:dsg}
With $\boldsymbol\lambda^0=\bm0$ and $\eta=\Lambda_{\max}/\sqrt{(K+1)(N+1)}$, the iterates satisfy
\begin{equation}\label{eq:dsg-rate}
\max_{0\le n\le N}D(\boldsymbol\lambda^n)\ \ge\
D^*-\frac{\Lambda_{\max}\sqrt{K+1}}{\sqrt{N+1}} .
\end{equation}
The running average $\bar{\boldsymbol\lambda}^N$ obeys the same bound.
\end{proposition}

\begin{proof}[Proof sketch]
Standard projected subgradient analysis~\cite[Sec.~3.2.3]{Nesterov2018Lectures} on the compact convex set $\Lambda$ (Lemma~\ref{lem:lambda-bound}), using non-expansiveness of the projection, $\|\bm g^n\|_2\le\sqrt{K+1}$, and $\|\boldsymbol\lambda^0-\boldsymbol\lambda^*\|_2 \le\Lambda_{\max}$.
\end{proof}

Two remarks on implementation. The inner MDP is communicating but need not be unichain, since deterministic policies can freeze estimates, so plain relative value iteration is not guaranteed to converge; the inner solve should use multichain policy iteration or the average-cost linear program~\cite[Ch.~9]{Puterman1994}. A primal solution is recovered by solving the program of Theorem~\ref{thm:mixture} restricted to the columns $\{(\pi^{\boldsymbol\lambda^n},R^n)\}_{n\le N}$. With $\bar W$ sweeps per inner solve, $N$ outer iterations cost $O(N\bar W|\mathcal{S}|^2|\mathcal{E}|)$.

\section{Numerical Results}\label{sec:numerical}
\subsection{Setup}\label{sec:num-setup}

We consider a system with $M=3$ sources monitored by $K=2$ sensors sharing the TDMA uplink. Sensor 1 is reliable but delivers one slot late, $(p^1_s,d_1)=(0.9,1)$, and covers all three sources; sensor 2 is faster but less reliable, $(p^2_s,d_2)=(0.55,0)$, and covers sources~2 and~3 only, thus source~1 is reachable through a single edge. In each slot the scheduler activates one of them or remains idle.

Each source evolves as a three-state Markov chain, for simplicity we write $Q(a)$ for the chain with self-transition probability $a$ and uniform off-diagonal probabilities, which is slowly evolving as described in Section~\ref{sec:system} exactly when $a>1/2$. We take $Q^1 = Q(0.9)$, $Q^2 = Q(0.2)$ and $Q^3 = Q(0.8)$. Sources~1 and~2 share the asymmetric CAE matrix $\delta =
\begin{bmatrix}
0 & 10 & 20\\
20 & 0 & 10\\
30 & 20 & 0
\end{bmatrix}$. Source 3 carries a uniform mismatch cost of $15$. We set $\omega_m=1$ without loss of generality. The budgets are $F_{\max}=0.5$ and $(F_{\max,1},F_{\max,2})=(0.35,0.30)$.

\begin{figure}[htp]
\centering
\includegraphics[width=0.7\columnwidth]{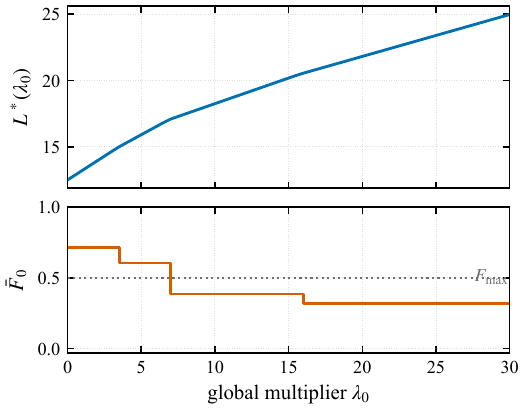}
\caption{Inner Lagrangian value and the corresponding transmission frequency.}
\label{fig:pwlc}
\end{figure}

\subsection{Piecewise-linear structure and randomization}\label{sec:num-pwlc}

Fig.~\ref{fig:pwlc} illustrates the structural properties established in Theorem~\ref{thm:pwlc}, $L^*$ is non-decreasing, concave, and piecewise linear up to numerical tolerance, and every slope of the envelope coincides with a plateau of the transmission frequency, which is the supergradient in~\eqref{eq:superdiff}.

The budget $F_{\max}=0.5$ falls strictly between two frequency plateaus. No deterministic inner policy is therefore both optimal and budget-tight, and the constrained optimum must randomize between deterministic policies, which is the case covered by Theorem~\ref{thm:mixture}. Projected dual subgradient ascent reaches a duality gap of $3.7\%$ after $40$ outer iterations and $0.05\%$ after $200$, consistent with Proposition~\ref{prop:dsg}.

\begin{figure}[htp]
\centering
\includegraphics[width=0.7\columnwidth]{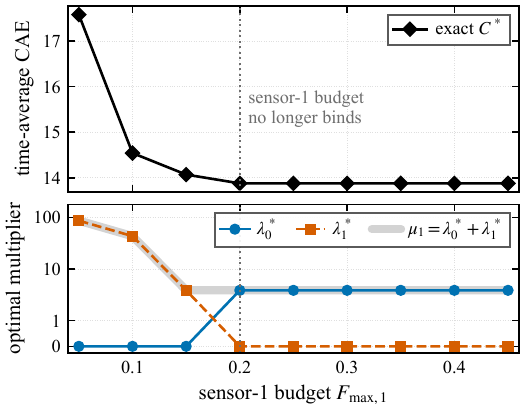}
\caption{Results of the sensor-1 budget $F_{\max,1}$.}
\label{fig:caps}
\end{figure}

\subsection{Effective transmission costs under a per-sensor budget}\label{sec:num-caps}

Fig.~\ref{fig:caps} shows how the optimal average cost and optimal Lagrange multipliers vary with the sensor-1 transmission budget $F_{\max,1}$. The exact optimum decreases until the budget reaches the level a budget-free optimum would use, $F_{\max}-F_{\max,2}=0.2$, and is flat beyond it. At that point the multipliers change roles as Proposition~\ref{prop:effective-price} requires: $\lambda^*_1$ drops from $3.86$ to $0$ while $\lambda^*_0$ rises from $0$ to $3.86$, so $\mu_1=\lambda^*_0+\lambda^*_1$ is continuous. The sign of $\kappa$ predicts which multiplier vanishes: $\lambda^*_0=0$ while $\kappa<0$ and $\lambda^*_1=0$ once $\kappa\ge0$. Only $\mu_1$ is determined by the inner problem; the split between $\lambda_0$ and $\lambda_1$ is fixed by which constraints bind.

\section{Conclusion}\label{sec:conclusion}

We studied semantic-aware remote estimation with overlapping sensor coverage under global and per-sensor transmission constraints. We established that an optimal policy can be implemented by randomizing once among at most $K+1$ stationary deterministic policies and characterized the
piecewise-linear concave Lagrangian value. Numerical results confirmed these structural properties and illustrated the performance of the projected dual subgradient method. Future work will focus on scalable scheduling for larger systems.

\bibliographystyle{ACM-Reference-Format-SEQ}
\bibliography{references}

\end{document}